\documentclass[11pt]{article}
\usepackage{amsmath,amssymb}
\usepackage{graphicx}
\usepackage{array}
\usepackage[ruled,vlined]{algorithm2e}
\usepackage{multirow}
\usepackage{hhline}

\newtheorem{theorem}{Theorem}[section]

\newtheorem{lemma}[theorem]{Lemma}

\numberwithin{equation}{section}
\newenvironment{proof}{\noindent\textbf{Proof\ }}{\hspace*{\fill}$\Box$\medskip}

\begin{document}

\title{An improvement on the Louvain algorithm using random walks}
\author{DO Duy Hieu and PHAN Thi Ha Duong
\\
\\
 Institute of Mathematics\\ Vietnam Academy of Science and Technology}
\date{\empty}
\maketitle
 Email: ddhieu@math.ac.vn (Do Duy Hieu),
phanhaduong@math.ac.vn (Phan Thi Ha Duong).
\begin{abstract}
We will present improvements to famous algorithms for community detection, namely Newman’s spectral method algorithm and the Louvain algorithm. The Newman algorithm begins by treating the original graph as a single cluster, then repeats the process to split each cluster into two, based on the signs of the eigenvector corresponding to the second-largest eigenvalue. Our improvement involves replacing the time-consuming computation of eigenvalues with a random walk during the splitting process.

The Louvain algorithm iteratively performs the following steps until no increase in modularity can be achieved anymore: each step consists of two phases, phase 1 for partitioning the graph into clusters, and phase 2 for constructing a new graph where each vertex represents one cluster obtained from phase 1. We propose an improvement to this algorithm by adding our random walk algorithm as an additional phase for refining clusters obtained from phase 1. It maintains a complexity comparable to the Louvain algorithm while exhibiting superior efficiency. To validate the robustness and effectiveness of our proposed algorithms, we conducted experiments using randomly generated graphs and real-world data.
\end{abstract}

\section{Introduction}
Research on community detection in networks is an essential field within network science, with a wide array of applications in computer science and various other scientific disciplines \cite{L18, L29, L38}. Consequently, numerous research efforts from scientists have employed various methodological approaches. Among these, two algorithms garnering significant attention are Newman's spectral method \cite{main} and the Louvain algorithm \cite{Louvain}. Therefore, numerous extensions and improvements have been made to these algorithms \cite{minh, FastLouvain, leiden, lei15}.

Random walk is a focal point of interest in network community research, as it helps elucidate the characteristics of vertices belonging to the same or different communities \cite{latapy, overlap}. This paper will employ random walks to improve Newman's spectral method and the Louvain algorithm.

\subsection{Newman’s Spectral Method}

In \cite{main}, Newman proposed a spectral method for detecting network communities. The initial algorithm constructs a normalized Laplacian matrix as follows:
\[ \mathbf{L} = D^{-1/2} A D^{-1/2} \]
Here, the matrix \( A \) is the adjacency matrix, and \( D \) is the diagonal matrix with elements equal to the vertex degrees \(D_{ii} = d_i\).

Assuming \(V = \{1, 2, \ldots, n\}\) represents the set of graph vertices. The algorithm then classifies the graph into two communities based on the eigenvectors corresponding to eigenvalues greater than 1 of the matrix \( \mathbf{L} \). Typically, the second eigenvector is chosen (assuming \(v_\beta = (v^1_\beta, v^2_\beta, \ldots, v^n_\beta)\)), and specifically, if \(v^i_\beta \geq 0\), vertex \(i\) belongs to community \(C_1\); otherwise, vertex \(i\) belongs to community \(C_2\).

Subsequently, the algorithm generates two subgraphs \(G_1\) and \(G_2\) corresponding to communities \(C_1\) and \(C_2\), and then repeats the process. The algorithm stops when the partitioning of the graph into communities \(C_1\) and \(C_2\) no longer increases the modularity value (modularity will be introduced in the next section).

In this paper, we also introduce a graph partition method similar to the one mentioned earlier. However, we employ a random walk at each step instead of using eigenvectors, thereby reducing our algorithm's computational complexity. Additionally, we establish a connection with the eigenvalues and eigenvectors of the normalized Laplacian matrix. This connection reveals that our algorithm produces clustering results equivalent to Newman's algorithm when the number of random walk steps is sufficiently large.

\subsection{Louvain Algorithm}

The Louvain algorithm \cite{Louvain} stands out for its simplicity and elegance. It optimizes a quality function, such as Modularity or CPM, through two primary phases: \\
\textbf{Phase 1:} Nodes assess potential relocation to neighboring communities by maximizing Modularity increase using the formula:
\begin{equation}\label{Qic}
    \Delta Q_{i,C_j} = \left[\frac{\Sigma_{in}+2 k_{i,in}}{2m} - \left(\frac{\Sigma_{tot}+k_i}{2m}\right)^2\right] - \left[\frac{\Sigma_{in}}{2m} - \left(\frac{\Sigma_{tot}}{2m}\right)^2 - \left(\frac{k_i}{2m}\right)^2\right]
\end{equation}  
Here, $\Delta Q_{i, C_j}$ signifies Modularity change upon placing vertex $i$ into community $C_j$, $\Sigma_{in}$ is the sum of weights of links within the transitioning community, $\Sigma_{tot}$ is the sum of weights of links to nodes in the transitioning community, $k_i$ is the weighted degree of $i$, $k_{i, in}$ is the sum of weights of links between $i$ and other nodes in the transitioning community, and $m$ is the sum of weights of all links in the network.\\
\textbf{Phase 2:} Consolidates nodes within the same community to form a new network. Self-loops denote intra-community links, and weighted edges represent inter-community connections. 

The algorithm iterates these phases until Modularity ceases to increase.

\subsection{Random walk on graphs}
Before delving into the next section, let's revisit the concept of a random walk on a graph and some associated knowledge. Consider an undirected and connected graph $G = (V, E)$ with sets of vertices ($V$) and edges ($E$). Let $|V|$ be denoted as $n$, $|E|$ as $m$, and the adjacency matrix of $G$ as $A$. In this context, $A_{ij} = 1$ if vertices $i$ and $j$ are connected (linked by an edge), and $A_{ij} = 0$ otherwise. The degree $d(i) = \sum_{j} A_{ij}$ of a vertex $i$ represents the Number of its neighbors, including itself. For simplicity, this paper focuses on unweighted graphs. Nevertheless, it is straightforward to extend the results to weighted graphs, where $A_{ij} \in \mathbb{R}^+$ instead of $A_{ij} \in {0,1}$, showcasing the versatility of this approach.

Now, consider a random walk $X = X_0, X_1, \ldots, X_t, \ldots$ on graph $G$ (refer to \cite{30, 4} for a comprehensive presentation). At each step $t$, the walker moves to a vertex randomly and uniformly selected from its neighbors. Consequently, the sequence of visited vertices forms a Markov chain, with states representing the graph's vertices. At each step, the transition probability from vertex $i$ to vertex $j$ is given by $P_{ij} = \frac{A_{ij}}{d(i)}$. This definition establishes the transition matrix $P$ for the random walk process.

It's evident that $P = D^{-1} A$, where $D$ is the diagonal matrix of degrees ($D_{ii} = d(i)$ and $D_{ij} = 0$ for $i \neq j$). Finally, the information about vertex $i$ encoded in $P^t$ resides in the $n$ probabilities $(P^t_{ik}){1\leq k\leq n}$, equivalent to the $i$-th row of matrix $P^t$ denoted by $P^t_{i\bullet}$.  We assume that $G$ is
connected.
 According to the convergence theorem for finite Markov chains, the associated transition matrix $P$ satisfies
$\lim_{k \rightarrow \infty}P =P_{\infty}$, where $(P_{\infty})_{ij} = \phi_j$, the $j$-th component of the unique stationary distribution $\phi=(\phi_1,\phi_2,...,\phi_n)$, note that $\phi_i=d(i)/\sum d(j)$ (see in \cite{latapy}).
This paper also represents $X_{i\bullet}$ as the $i$-th row of an arbitrary matrix $X$.

\subsection{Modularity}
 The modularity $Q$ introduced in \cite{La32,La33}, which relies on the fraction of edges $e_C$ inside community $C$ and the fraction of edges $a_C$ bound to community $C$ :
\begin{equation}\label{QQ}
Q(\mathcal{P})=\sum_{C \in \mathcal{P}} (e_C-a_C^2)
\end{equation}
From this, we can consider the Modularity corresponding to each clustering for the graph as the total Modularity of each community. From there, with the $C$ community, we have the Modularity corresponding to it:
\begin{equation}
Q(C,G)= e_C-a_C^2
\end{equation}

\subsection{Our contribution}

In this paper, we first employ a random walk strategy to introduce a novel method for community detection through graph partitioning, akin to the spectral method in \cite{main}, known as the Random Walk Graph Partition Algorithm. This algorithm exhibits lower computational complexity and greater efficiency than the spectral algorithm in \cite{main}.

Furthermore, considering the Random Walk Graph Partition Algorithm as the 1.5th step of the Louvain algorithm, we propose a new algorithm named the Random Walk Graph Partition Louvain Algorithm. This algorithm maintains computational complexity equivalent to the Louvain algorithm but achieves higher efficiency, especially for graphs with unclear community structures.

Additionally, we conduct experiments on both randomly generated and real data to demonstrate the reasonability and effectiveness of our proposed algorithms.

\section{Random Walk Graph Partition Algorithm}
Newman's spectral approach is based on spectral analysis of the normalized Laplacian matrix $\mathbf{L}$. This approach is quite good because the Laplacian matrix is a matrix that has many properties of graphs and is the subject of many people's research. Furthermore, this algorithm is also quite simple. However, we need to find eigenvalues and eigenvectors with great computational complexity to implement the algorithm. Therefore, instead of directly analyzing the matrix Laplacian, we will study it through a random walk. We note that the transition matrix of the random walk $P$ is similar to the matrix $\mathbf{L}$ ($\mathbf{L} = D^{-\frac{1}{2}} A D^{-\frac{1}{2}} = D^{\frac{1}{2}} P D^{-\frac{1}{2}}$), so they have the same eigenvalues, and the eigenvectors are related.

Now, we will analyze the relationship between the spectral approach and random walk. First, we recalled the following lemma.
\begin{lemma}(\cite[Lemma 1]{latapy})\label{lm1}
The eigenvalues of the matrix $P$ are real and satisfy:
\begin{equation}
1 = \lambda_1 > \lambda_2 \geq \ldots \geq \lambda_n > -1.
\end{equation}
Moreover, there exists an orthonormal family of vectors $(s_\alpha)_{1\leq \alpha \leq n}$ such that each vector
$v_\alpha = D^{-1/2}s_\alpha$ and $u_\alpha = D^{1/2}s_\alpha$ are respectively a right and a left eigenvector associated to
the eigenvalue $\lambda_\alpha$:
\begin{equation*}
    \forall \alpha,\, Pv_\alpha = \lambda_\alpha v_\alpha \,\, \mbox{ and }  P^Tu_\alpha = \lambda_\alpha u_\alpha
\end{equation*}
$$ \forall \alpha, \forall \beta,\,\, v^T_\alpha u_\beta = \delta_{\alpha \beta}$$
\end{lemma}

 This paper only considers the case $\lambda_2 \neq \lambda_3$. We have the following theorem from Lemma \ref{lm1}.
\begin{theorem}\label{thm1}
Let $i$ be any vertex of the graph $G$. Then we have the $j-$th component of vector $  P^t_{i \bullet} - \phi$ and the $j-$th component of vector $s_2$ have the same sign for all $j=1,2,...,n$ or opposite signs for all $j=1,2,...,n$, where $s_ 2$ is the eigenvector corresponding to the second eigenvalue of the normalized Laplacian matrix $\mathbf{L}$.
\end{theorem}
\begin{proof}
Lemma \ref{lm1} makes it possible to write a spectral decomposition of the matrix $P:$
\begin{equation}
    P= \sum_{\alpha=1}^n \lambda_\alpha v_\alpha u^T_\alpha \, \mbox{ and } \, P^t = \sum_{\alpha=1}^n \lambda^t_\alpha v_\alpha u^T_\alpha.
\end{equation}
 It follows that 
\begin{equation}\label{pt2.6}
   P^t_{ij} = \sum_{\alpha=1}^n \lambda^t_\alpha v_\alpha(i) u^T_\alpha(j) \, \mbox{ and }  P^t_{i \bullet} = \sum_{\alpha=1}^n \lambda^t_\alpha v_\alpha(i) u_\alpha.
\end{equation}
When $t$ tends towards infinity, all the terms $\alpha \geq 2$ vanish. It is easy to show that the
first right eigenvector $v_1$ is constant. By normalizing we have $\forall i, \, v_1(i)= \frac{1}{\sum_k d_k} $ and $\forall j, \, u_1(j)= \frac{d_j}{\sum_k d_k}.$  
Therefore, we have 
\begin{equation}\label{pt2.7}
\lim_{t\rightarrow \infty} P^t_{ij} = \lim_{t\rightarrow \infty}  \lambda^t_\alpha v_\alpha(i) u^T_\alpha(j) =v_1(i) u^T_1(j) = \frac{d_j}{\sum_k d_k}=\phi_j.\end{equation}
From (\ref{pt2.6}), (\ref{pt2.7}) and $\lambda_1=1$, we have
\begin{equation}
    P^t_{i \bullet} - \phi = 
    \sum_{\alpha=2}^n \lambda^t_\alpha v_\alpha(i) u_\alpha =\lambda^t_2 v_2(i) u_2+\lambda^t_3 v_3(i) u_3+\ldots+\lambda^t_n v_n(i) u_n,
\end{equation}
this is equivalent to
\begin{equation}\label{pttk}
    P^t_{i \bullet} - \phi  =\lambda^t_2 \left(v_2(i) u_2+  \frac{\lambda^t_3}{\lambda^t_2} v_3(i) u_3+\ldots+\frac{\lambda^t_n}{\lambda^t_2} v_n(i) u_n\right).
\end{equation}
On the other hand, from Theorem \ref{thm1}, we have $u_\alpha=D^{-1/2}s_\alpha$. From there, it follows.
\begin{equation}\label{pttk1}
    P^t_{i \bullet} - \phi  =\lambda^t_2 D^{-1/2} \left(v_2(i) s_2+  \frac{\lambda^t_3}{\lambda^t_2} v_3(i) s_3+\ldots+\frac{\lambda^t_n}{\lambda^t_2} v_n(i) s_n\right).
\end{equation}
From Theorem \ref{thm1} and $\lambda_2 \neq \lambda_3$,  we have $|\lambda_\alpha/\lambda_2| < 1$ with $3\leq \alpha \leq n$. 
Therefore,
when $t$ tends towards infinity, all the terms in \ref{pttk} with $\alpha \geq 3$ vanish. From there we have the $j-$th component of vector $ P^t_{i \bullet} - \phi$ and the $j-$th component of vector $\lambda^t_2 v_2(i) s_2$ having the same sign for $t$ large enough and for all $j=1,2,...,n$. Hence the conclusion of the theorem.
Furthermore, from (\ref{pttk1}), we deduce this conclusion holds for all $i=1,2,...,n$.
\end{proof}

From Theorem \ref{thm1}, we observe that clustering based on random walk and spectral analysis is the same. Therefore, based on a random walk, we can propose a Random Walk Graph Partition Algorithm \ref{alg:RWGP1} as follows.

\begin{algorithm}[H]
    \SetAlgoLined
    \KwIn{Graph $G$, $C \subset V(G)$, $Q = \emptyset$, $t$}
    \KwOut{Final list of clusters $Q$}

    \textbf{Phase 1:}

    $C_1 = \emptyset$, $C_2 = \emptyset$\;
    
    Create induced graph $G'$ from $C$, and select any vertex $i_0$ in cluster $V(G')$\;
    
    Calculate $ P^t_{i_0\bullet} - \phi = (P^t_{i_01} - \phi_1, P^t_{i_02} - \phi_2, ..., P^t_{i_0n} - \phi_n)$ in $G'$\;
    
    \For{each vertex $j$}{
        \If{$P^t_{i_0j} - \phi_j \geq 0$}{
            Add $j$ to cluster $C_1$\;
        }
        \Else{
            Add $j$ to cluster $C_2$\;
        }
    }

    \textbf{Phase 2:}

    \If{$C_1$, $C_2$ are non-empty \textbf{and} $Q(C_1, G) + Q(C_2, G) > Q(C, G)$}{
        Apply \textbf{Phase 1} with $C = C_1$ and apply \textbf{Phase 1} with $C = C_2$\;
    }
    \Else{
        $Q = Q \cup \{C\}$\;
    }

    \caption{Random Walk Graph Partition Algorithm 1}
    \label{alg:RWGP1}
\end{algorithm}

In Algorithm \ref{alg:RWGP1}, in phase 1, we need to compute $P^t_{i_0\bullet}$. The computational complexity for calculating $P^t_{i_0\bullet}$ is $O(t m)$, where $m$ is the number of edges in the graph $G$ (see \cite{latapy}). In phase 2, we need to compute $Q(C, G)$. To calculate $Q(C, G)$, we need to count the number of edges in $C$ ($e_C$) and the number of edges connected to $C$ ($a_C$), so the computational complexity of phase 2 does not exceed $O(m)$. Therefore, the computational complexity of each iteration is $O(t m) + O(m) = O(t m)$. Assuming the number of communities in the graph $G$ is $k$, the computational complexity of Algorithm \ref{alg:RWGP1} is $O(tkm)$.

Although the case $\lambda_2=\lambda_3$ rarely occurs; if it does happen or $\lambda_2$ is very close to $\lambda_3$, our Random Walk Graph Partition Algorithm 1 will no longer be accurate. Therefore, to improve the effectiveness of our algorithm in these cases, after dividing cluster $C$ into two clusters $C_1, C_2$, we will add an adjustment step. Specifically, we will review all vertices to see if they deserve to be in the current cluster or if they should be moved to another cluster based on maximizing modularity. From there, we propose the following algorithm. In this part, we shall use the notation $C_i$ to represent the community that includes node $i$ and $C_{\overline{i}}$ for the community that does not include node $i$.

\begin{algorithm}[H]
    \SetAlgoLined
    \KwIn{Graph $G$, $C \subset V(G)$, $Q = \emptyset$, $t$}
    \KwOut{Final list of clusters $Q$}

    \textbf{Phase 1:}

    $C_1 = \emptyset$, $C_2 = \emptyset$\;
    
    Create induced graph $G'$ from $C$, and select any vertex $i_0$ in cluster $V(G')$\;
    
    Calculate $ P^t_{i_0\bullet} - \phi = (P^t_{i_01} - \phi_1, P^t_{i_02} - \phi_2, ..., P^t_{i_0n} - \phi_n)$ in $G'$\;
    
    \For{each vertex $j$}{
        \If{$P^t_{i_0j} - \phi_j \geq 0$}{
            Add $j$ to cluster $C_1$\;
        }
        \Else{
            Add $j$ to cluster $C_2$\;
        }
    }
    Count.number.move=1
    
    \While{Count.number.move $\neq$ 0}{
        Count.number.move=0\;
        \For{$i$ in $V(G')$}{
            \If{$\Delta Q_{i,C_i\setminus \{i\}} < \Delta Q_{i,C_{\overline{i}}} $}{
                Add $i$ to $C_{\overline{i}}$\;
                
                Count.number.move = Count.number.move + 1\;
            }
        }
    }

    \textbf{Phase 2:}

    \If{$C_1$, $C_2$ are non-empty \textbf{and} $Q(C_1, G) + Q(C_2, G) > Q(C, G)$}{
        Apply \textbf{Phase 1} with $C = C_1$ and apply \textbf{Phase 1} with $C = C_2$\;
    }
    \Else{
        $Q = Q\cup \{C\}$\;
    }

    \caption{Random Walk Graph Partition Algorithm 2}
    \label{alg:RWGP2}
\end{algorithm}

The algorithm referred to as Algorithm \ref{alg:RWGP2} distinguishes itself from Algorithm \ref{alg:RWGP1} solely in the adjustment step, which involves computing $\Delta Q_{i,C}$. Regarding the while loop, it is worth mentioning that this loop only iterates a few times, and the computational complexity of calculating $\Delta Q_{i,C}$ does not surpass $O(m)$. Consequently, the computational complexity of Algorithm \ref{alg:RWGP2} is equivalent to that of Algorithm \ref{alg:RWGP1}.

\section{Random Walk Graph Partition Louvain algorithm}


The Louvain algorithm \cite{Louvain} is a very famous algorithm not only because of its fast calculation speed but also because of its high algorithm accuracy. It iteratively performs the following steps until no increase in modularity can be achieved anymore. It's crucial that each step consists of two phases: phase 1 for partitioning the graph into clusters, and phase 2 for constructing a new graph where each vertex represents one cluster obtained from phase 1.

However, this algorithm could be less effective when the network has an unclear community structure. In the paper \cite{leiden}, Leiden proposed one of the most prominent improvements of the Leiden algorithm - the author added a phase of fine-tuning local communities after the first phase of the Louvain algorithm.

The Leiden algorithm fine-tuning local communities doesn't work effectively in cases where the community could be more opaque. Therefore, this paper will propose an algorithm with a more effective method of fine-tuning local communities. More precisely, we will use the Random Walk Graph Partition Algorithm mentioned above to perform the refinement for each cluster obtained from Phase 1. Consequently, we propose a new algorithm named the \textbf{Random Walk Graph Partition Louvain Algorithm} (or \textbf{RWGP-Louvain Algorithm} for short).

\begin{algorithm}[H]
\SetAlgoLined
\KwIn{Network $G_{ori}$, $\mathcal{P}=\emptyset$, $t$}
\KwOut{FinalCommunities - Final list of communities}

$G=G_{ori}$

\textbf{Phase 1:}

\For{ $i$ in $V(G)$}{
    $C_i=\{i\}$,\, $\mathcal{P}=\mathcal{P}\cup C_i$\;
}

\While{some nodes are moved}{
    \For{ $i$ in $V(G)$}{
        \For{neighboring community $C_j$ of $i$}{
            Calculate $\Delta Q_{i,C_j}$ according to Formula \ref{Qic}\;
            
            Add $i$ to the community $C_{i_0}$ with maximizing $\Delta Q_{i,C_{j_0}}$\;
        }
    }
}
\For{$C_j$ in $\mathcal{P}$}{
\If{$C_j=\emptyset$}{
    $\mathcal{P}=\mathcal{P}\setminus \{C_{j}\}$\;
}
}

\textbf{Phase 2:}  

If the clusters in $\mathcal{P}$ contain supernodes, then return the clusters containing the nodes of the original graph $G_{ori}$.

\For{$C_j$ in $\mathcal{P}$}{
Apply Random Walk Graph Partition Algorithm 1 or Random Walk Graph Partition Algorithm 2 with $G=G_{ori},\, C=C_j$, $Q=\emptyset $;

$\mathcal{P}=\mathcal{P}\setminus \{C_j\} \cup Q$
}

\textbf{Phase 3:}

Create a new network $G'$ by consolidating nodes within the same community to \textbf{supernodes}\;

\For{each pair of nodes $u$ and $v$ in the same community}{
    Add a self-loop to the community node for $u$ and $v$\;
 
}

\For{each edge between nodes in different communities}{
    Add a weighted edge between the corresponding community nodes\;
}

We will repeat \textbf{phases} 1 to 3 until modularity no longer increases.\;

\caption{Random Walk Graph Partition Louvain Algorithm}
\label{alg:RWGP-louvain}
\end{algorithm}

In Algorithm \ref{alg:RWGP-louvain}, compared to the Louvain algorithm, each iteration of the algorithm involves applying the Random Walk Graph Partition Algorithm to $G=G_{ori}$ and $C=C_j$ for each $C_j \in \mathcal{P}$. For each $C_j$, the computational complexity is $O(t m_{C_j})$, where $m_{C_j}$ is the number of edges in cluster $C_j$. Consequently, the overall computational complexity for this part is $O(tm)$, where $m$ represents the number of edges in the original graph $G_{ori}$.

\section{Experiments}
Evaluating a community detection algorithm is difficult because one needs
some test graphs with already known community structure. A classical approach is
to use randomly generated graphs with given communities. Here, we will use this approach and
generate the graphs as follows. \\
\textbf{Planted l-partition model}:
The first generator model is the planted l-partition model
\cite{2-out}. By determining the number of groups $l$, the number
of vertices in each group $g$, and two probabilities of inter-cluster $p_{in}$
and intra-cluster $p_{out}$, we obtain one random graph with some
Property:
\begin{itemize}
\item The average degree of one vertex is $E\left[k\right]=p_{in}(g-1)+p_{out}g(l-1)$.
\item All communities have the same size.
\item All vertices have approximately the same degree because of each community.
It can be seen as one random graph proposed by Erdos and Renyi. Each pair of vertices is connected in those random graphs with equal probability
$p_{in}$ independent of other pairs.
\end{itemize}
\textbf{Gaussian random partition generator}:
The next generator graph model is Gaussian random partition generator \cite{2-out},
which overcomes the part disadvantage of the planted $ l-$ partition model
above the vertex degree distribution. Unlike the planted
$ l-$ partition model, the community size in the Gaussian random partition
generator is a random variable of Gaussian distribution. The parameters
What needs to be determined for the Gaussian random partition generator are:
\begin{itemize}
\item Number of vertices in the graph: $N$.
\item Mean of community's size: $m$ and variance of community's size: $\sigma$.
\item Edge probability of inter $p_{in}$ and intra-cluster $p_{out}$.
\end{itemize}

After clustering a graph, we need to evaluate its quality. Therefore, we will present a metric to evaluate the clustering quality next. 

\subsection{Evaluating metrics}
In our experiments, we will use two metrics to compare the algorithms. The first metric is to use modularity (formula \ref{QQ}). The second metric is that we use Normalized Mutual Information (NMI).
Normalized Mutual Information \cite{NMI} quantifies the similarity between true class labels $Y$ and predicted cluster assignments $C$. It is computed as:
\begin{equation}\label{NMI}
\mathrm{NMI}(Y, C)=\frac{2 \times \mathrm{MI}(Y, C)}{\mathrm{H}(Y)+\mathrm{H}(C)},
\end{equation}
where:
\begin{itemize}
    \item $\mathrm{MI}(Y, C)$ is the Mutual Information between $Y$ and $C$,
\item $\mathrm{H}(Y)$ and $\mathrm{H}(C)$ are the entropies of $Y$ and $C$ respectively.
\end{itemize}
The NMI values range from 0 to 1, with higher values indicating better clustering alignment with true class labels. It is a normalized measure commonly used in cluster evaluation.

\subsection{Experiments for Random Walk Graphs Partition Algorithm}\label{RWGP}
In this section, we conduct experiments to compare the effectiveness of our two algorithms, Random Walk Graphs Partition Algorithm 1 and Random Walk Graphs Partition Algorithm 2, with the algorithms Louvain \cite{Louvain} and Newman's Spectral Method \cite{main}. The experiments involve randomly generated graphs using the Gaussian random generator and Planted-l partition models. Given that this is not the primary focus of our paper, we will perform a limited set of experiments.

We conduct ten trials on randomly generated graphs for each experiment using either the Gaussian random generator or the Planted-l partition model. Subsequently, we calculate Modularity for the clustering results obtained from different algorithms.

Finally, we present graphical representations of the Modularity values corresponding to Random Walk Graphs Partition Algorithm 1 (RWGP1), Random Walk Graphs Partition Algorithm 2 (RWGP2), Louvain \cite{Louvain}, and Newman's Spectral Method \cite{main} (Newman).

\subsubsection{Experiments for Random Walk Graphs Partition Algorithm
on the random graph generated by the Gaussian random generator
model}
With the Gaussian random generator model, we will experiment on graphs with the number of vertices ranging from a few hundred to tens of thousands. For each type of graph, we will fix $p_{in}=0.7$, and we will explore $p_{out}$ values of $0.01$, and $0.03$, corresponding to graphs with clear community structure to graphs with unclear community structure. And we will set the variance of the community’s size $\delta =2.5$. 

\subsubsection*{Experiment 1: Experiment on the random graph generated by the Gaussian random generator model.}
Using the Gaussian random generator model with a number of vertices $N$, the mean of the community’s size $m$ is taken with a uniform distribution in the following corresponding intervals: $N \in [500;1000]$, $m \in [50;100]$. In implementing the Random Walk Graphs Partition Algorithm 1 and Random Walk Graphs Partition Algorithm 2, we set $t$ to $15$.
We present these results in  Figure \ref{Fig1}. 

\begin{figure}
\begin{centering}
\includegraphics[width=1\columnwidth]{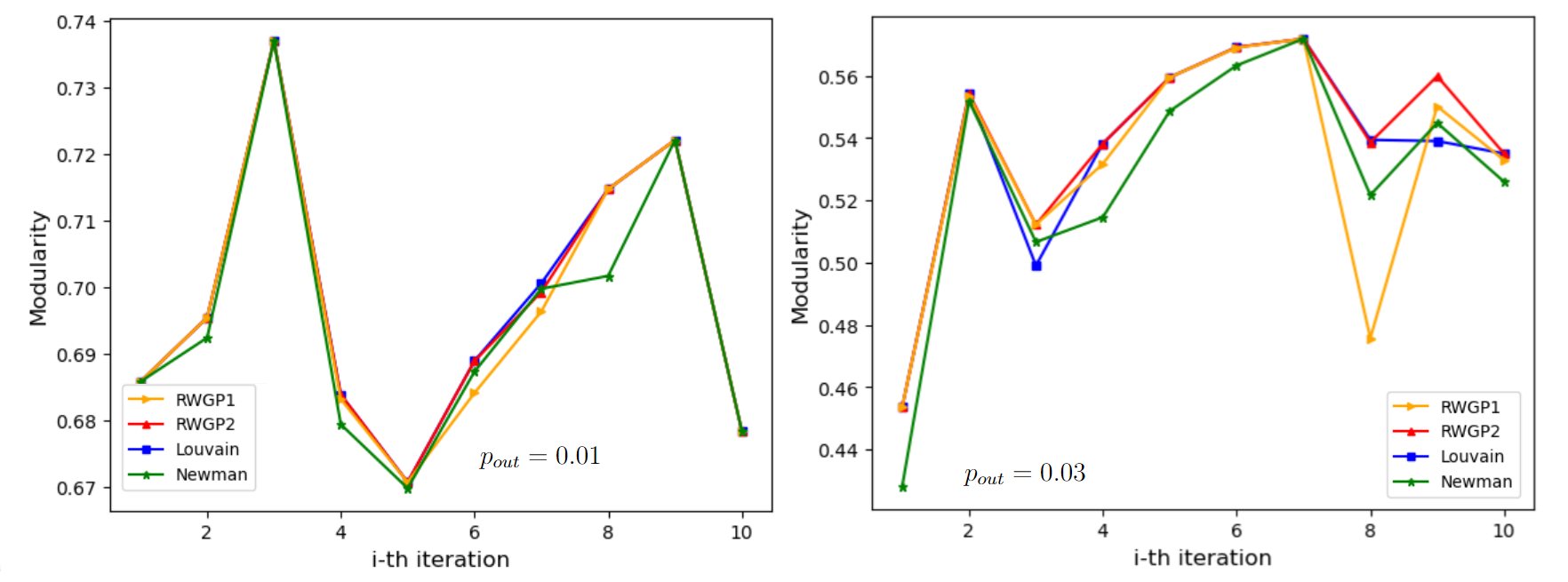}
\par\end{centering}
\caption{Modularity obtained in \textbf{Experiment 1} using Gaussian random generator model with $N \in [500;1000]$, $m \in [50;100]$.}\label{Fig1}
\end{figure}

\subsubsection*{Experiment 2: Experiment on the random graph generated by the Gaussian random generator model.}
We use the Gaussian random generator model with a number of vertices $N$, the mean of the community’s size $m$ taken with a uniform distribution in the following corresponding intervals: $N \in [1000;2000]$, $m \in [100;200]$. In implementing the Random Walk Graphs Partition Algorithm 1 and Random Walk Graphs Partition Algorithm 2, we set $t$ to $40$.
We present these results in  Figure \ref{Fig2}. 

\begin{figure}
\begin{centering}
\includegraphics[width=1\columnwidth]{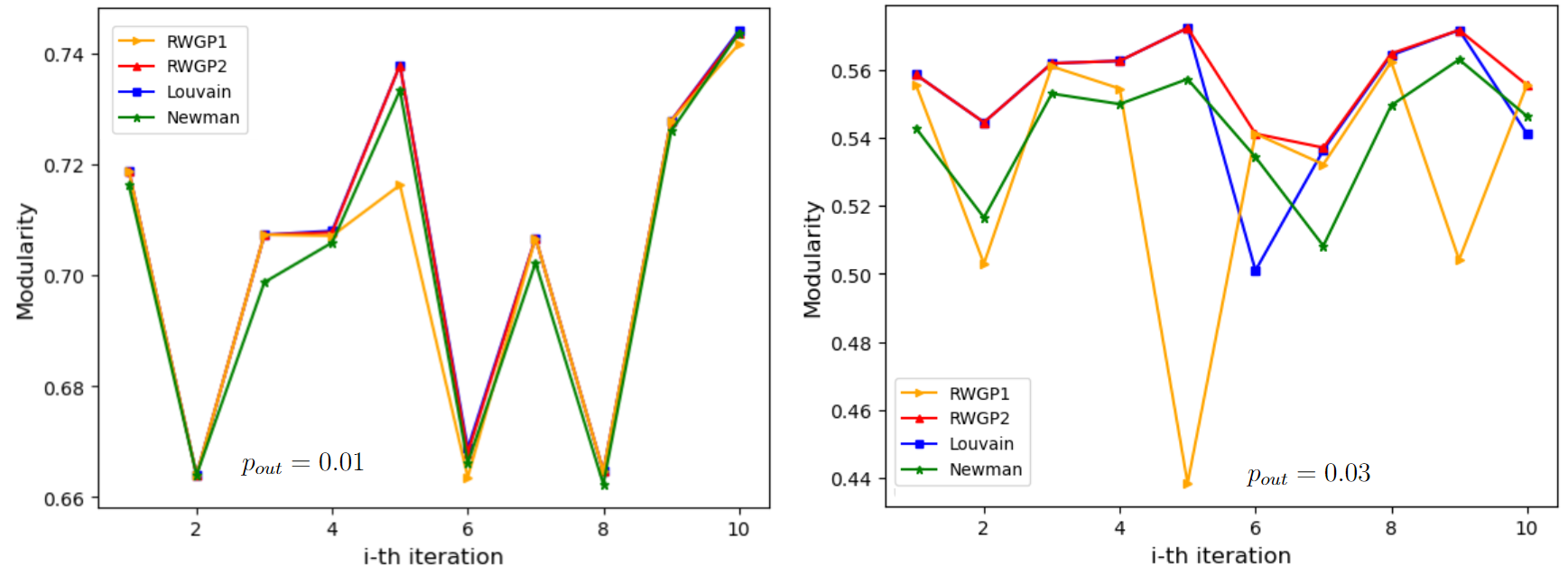}
\par\end{centering}
\caption{Modularity obtained in \textbf{Experiment 2} using Gaussian random generator model with $N \in [1000;2000]$, $m \in [100;200]$.}\label{Fig2}
\end{figure}

\subsubsection{Experiments for Random Walk Graphs Partition Algorithm on the random graph
generated by the Planted-l partition model}
With the Planted-l partition model, we also fix $p_{in}=0.7$ and explore $p_{out}$ values of $0.01$, and $0.03$, corresponding to graphs with clear community structure to graphs with unclear community structure. 
\subsubsection*{Experiment 1: Experiment on the random graph generated by the Planted-l partition
model.}
Using the Planted-l partition model with a number of communities $l$, the size of each community $g$ is taken with a uniform distribution in the following corresponding intervals: $g \in [3;5]$, $l \in [50;70]$. In implementing the Random Walk Graphs Partition Algorithm 1 and Random Walk Graphs Partition Algorithm 2, we set $t$ to $15$.
We present these results in  Figure \ref{Fig3}.

\begin{figure}
\begin{centering}
\includegraphics[width=1\columnwidth]{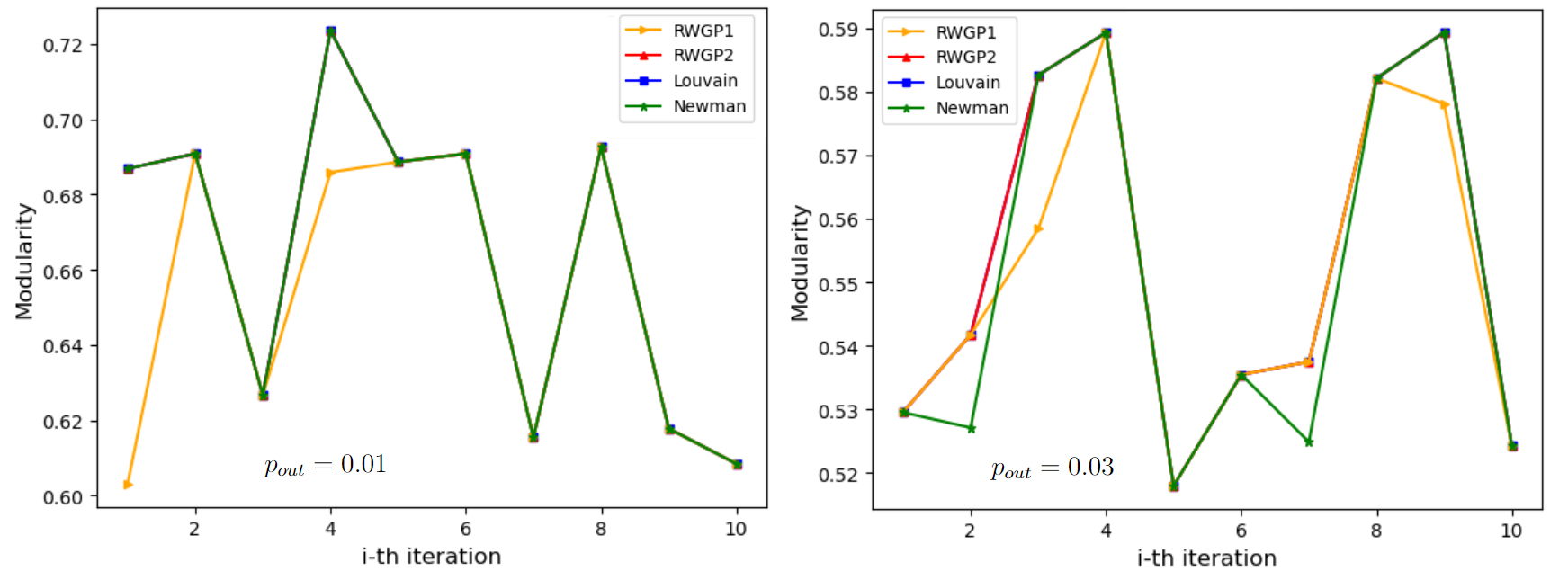}
\par\end{centering}
\caption{ Modularity obtained in \textbf{Experiment 1} using the Planted-l partition
model with $g \in [3;5]$, $l \in [50;70]$. }\label{Fig3}
\end{figure}

\subsubsection*{Experiment 2: Experiment on the random graph generated by the Planted-l partition
model.}
We use the Planted-l partition model with a number of communities $l$, the size of each community $g$ taken with a uniform distribution in the following corresponding intervals: $g \in [5;10]$, $l \in [100;200]$. In implementing the Random Walk Graphs Partition Algorithm 1 and Random Walk Graphs Partition Algorithm 2, we set $t$ to $40$.
We present these results in  Figure \ref{Fig4}.

\begin{figure}
\begin{centering}
\includegraphics[width=1\columnwidth]{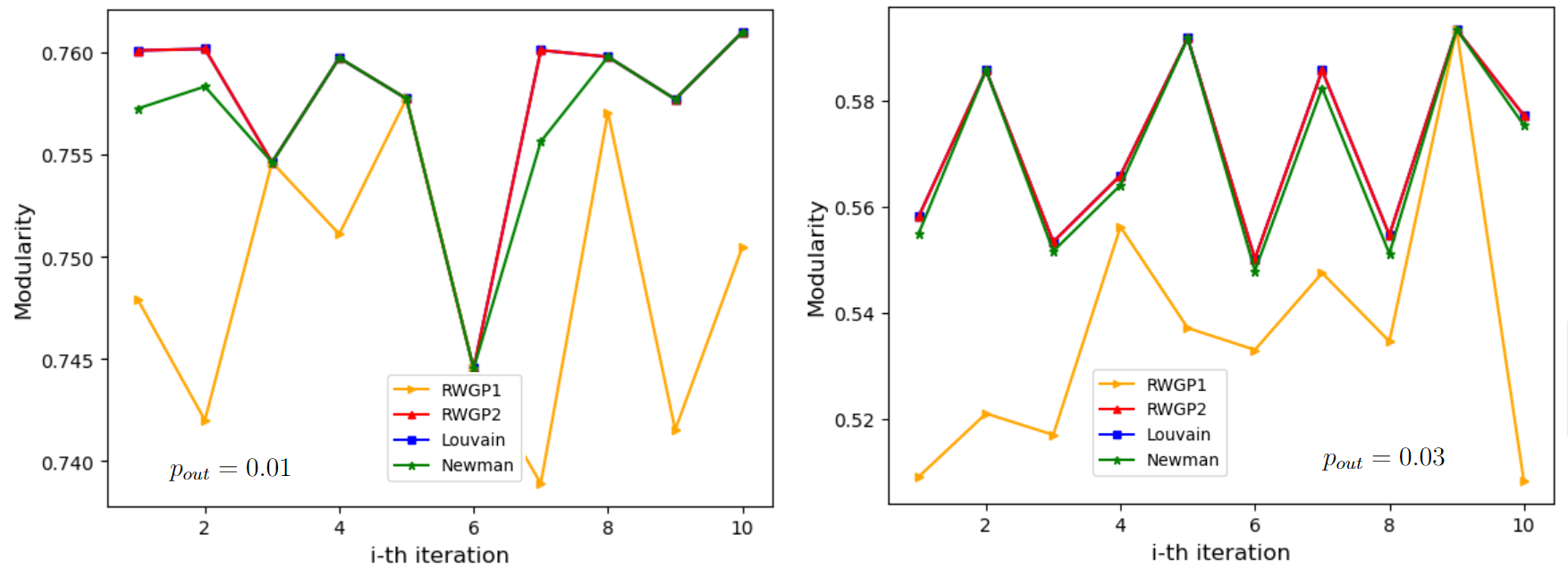}
\par\end{centering}
\caption{Modularity obtained in \textbf{Experiment 2} using the Planted-l partition
model with $g,l$ taken with $g \in [5;10]$, $l \in [100;200]$, and $p_{in}=0.7$. }\label{Fig4}
\end{figure}

\subsection{Experiments for Random Walk Graphs Partition Louvain Algorithm on graphs are randomly generate}\label{RWGPL}

In this section, we conduct experiments to compare the effectiveness of our proposed Random Walk Graphs Partition Louvain Algorithm (RWGP-Louvain) with existing algorithms, namely the original Louvain algorithm \cite{Louvain}, the Fast Louvain algorithm \cite{FastLouvain}, and the Leiden algorithm \cite{leiden}. The experiments involve randomly generated graphs using the Gaussian random generator and Planted-l partition models.

We perform ten trials on randomly generated graphs for each experiment utilizing either the Gaussian random generator or the Planted-l partition model. Subsequently, we calculate Modularity (use formula \ref{QQ}) for the clustering results obtained from different algorithms and calculate the NMI (use formula \ref{NMI}) between the clustering results obtained when applying the algorithms and the original clustering generated when creating the graph. 

Finally, we present the Modularity and NMI values corresponding to RWGP-Louvain, Louvain \cite{Louvain}, Fast Louvain \cite{FastLouvain}, and Leiden \cite{leiden} algorithms through graphical representations. We note that we apply the Random Walk Graph Partition Algorithm 2 in Phase 2 of the Random Walk Graph Partition Louvain Algorithm.

\subsubsection{Experiments for Random Walk Graphs Partition Louvain Algorithm on the random graph generated by the Gaussian random generator model}

With the Gaussian random generator model, we will experiment on graphs with the number of vertices ranging from a few hundred to tens of thousands. For each type of graph, we will fix $p_{in}=0.7$, and we will explore $p_{out}$ values of $0.01$, $0.03$, and $0.05$, corresponding to graphs with clear community structure to graphs with unclear community structure. And we will set the variance of the community’s size $\delta =2.5$.

\subsubsection*{Experiment 1: Random Graphs from the Gaussian Random Generator Model}
Using the Gaussian Random Generator Model, we set the number of vertices (\(N\)) and the mean community size (\(m\)) with a uniform distribution within \(N \in [500;1000]\) and \(m \in [20;30]\). For RWGP-Louvain Algorithm, \(t\) is set to \(15\).
The results are depicted in Figures \ref{Fig1a} and \ref{Fig1a1}. 
\begin{figure}
\begin{centering}
\includegraphics[width=1\columnwidth]{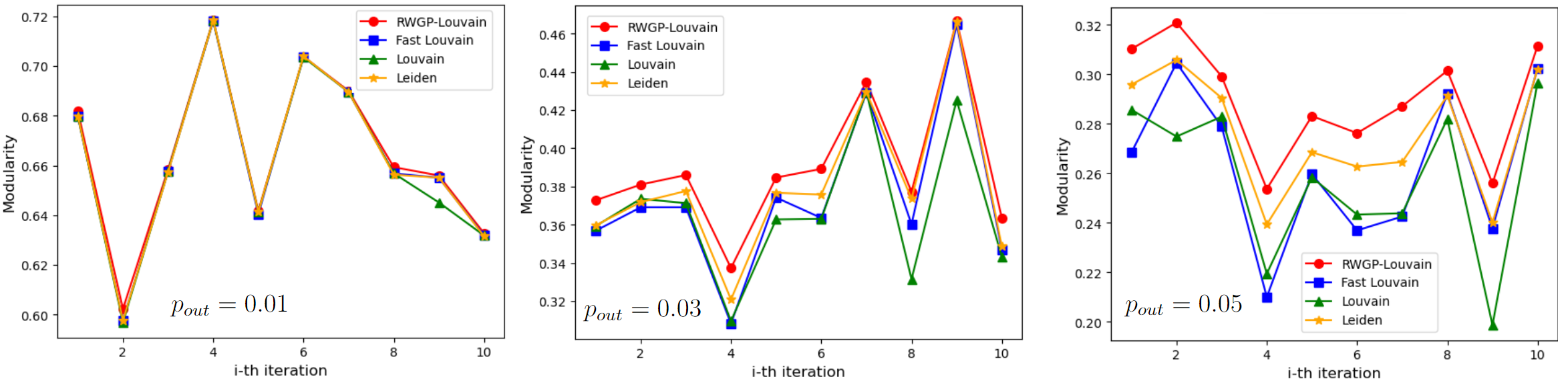}
\par\end{centering}
\caption{Modularity was observed in \textbf{Experiment 1} using the Gaussian random generator model with $N \in [500;1000]$ and $m \in [20;30]$.}\label{Fig1a}
\end{figure}

\begin{figure}
\begin{centering}
\includegraphics[width=1\columnwidth]{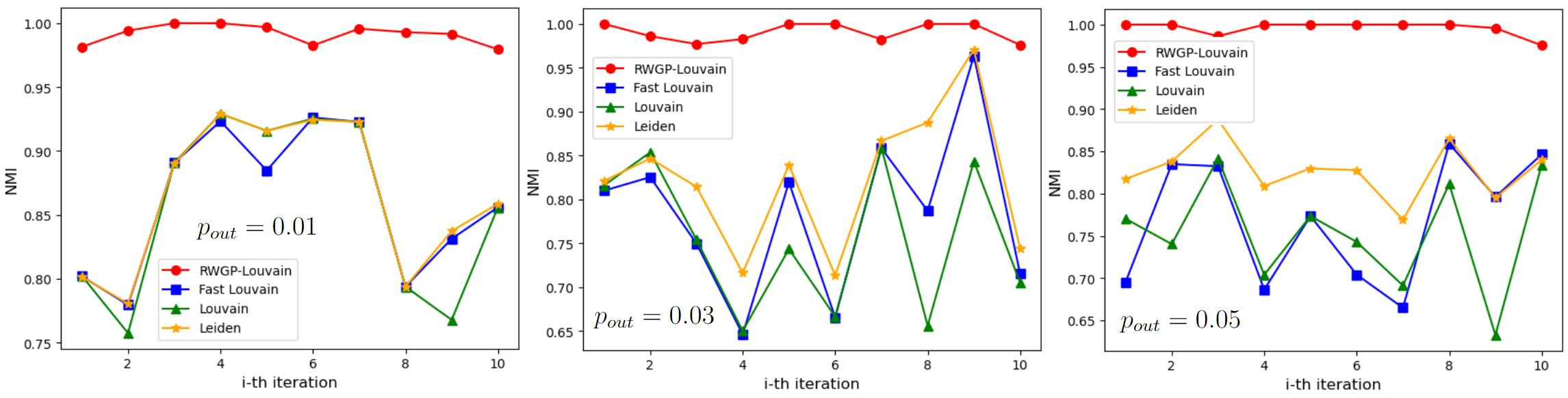}
\par\end{centering}
\caption{NMI was observed in \textbf{Experiment 1} using the Gaussian random generator model with $N \in [500;1000]$ and $m \in [20;30]$.}\label{Fig1a1}
\end{figure}

\subsubsection*{Experiment 2: Random Graphs from the Gaussian Random Generator Model}
Using the Gaussian Random Generator Model, we set the number of vertices (\(N\)) and the mean community size (\(m\)) with a uniform distribution within \(N \in [2000;4000]\) and \(m \in [50;70]\). For RWGP-Louvain Algorithm, \(t\) is set to \(25\). The results are illustrated in Figures \ref{Fig2a} and \ref{Fig2a1}.
\begin{figure}
\begin{centering}
\includegraphics[width=1\columnwidth]{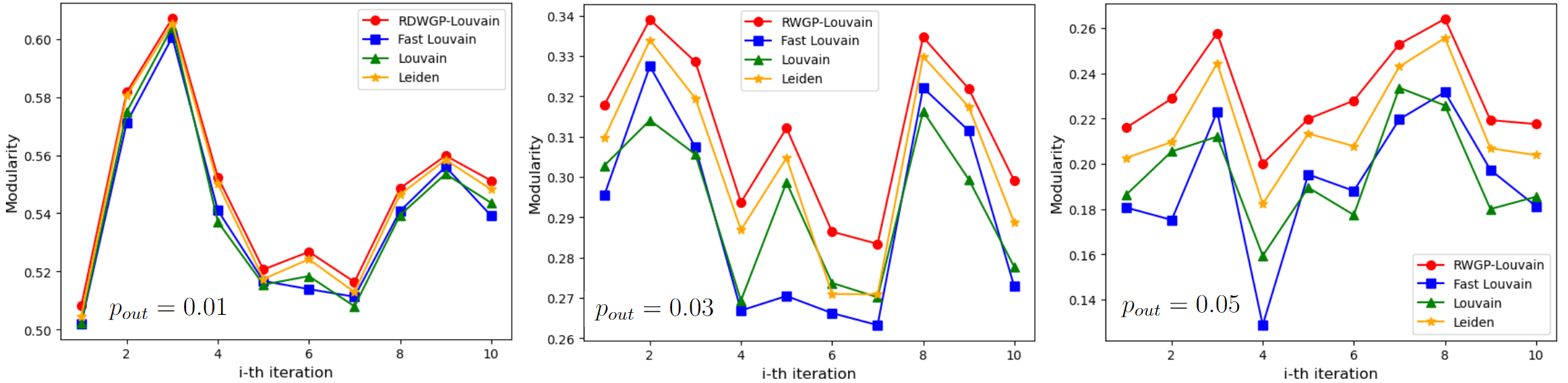}
\par\end{centering}
\caption{Modularity was observed in \textbf{Experiment 2} using the Gaussian random generator model with $N \in [2000;4000]$ and $m \in [50;70]$.}\label{Fig2a}
\end{figure}

\begin{figure}
\begin{centering}
\includegraphics[width=1\columnwidth]{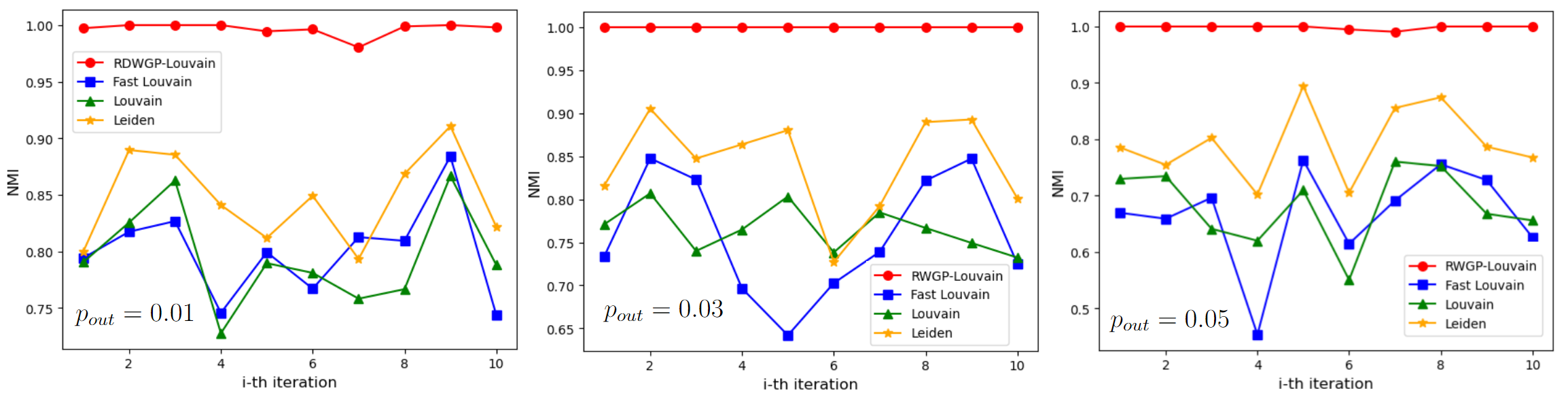}
\par\end{centering}
\caption{NMI was observed in \textbf{Experiment 2} using the Gaussian random generator model with $N \in [2000;4000]$ and $m \in [50;70]$.}\label{Fig2a1}
\end{figure}

\subsubsection*{Experiment 3: Random Graphs from the Gaussian Random Generator Model}
Using the Gaussian Random Generator Model, we set the number of vertices (\(N\)) and the mean community size (\(m\)) with a uniform distribution within \(N \in [4000;8000]\) and \(m \in [100;150]\). For RWGP-Louvain Algorithm, \(t\) is set to \(40\). The outcomes are presented in Figures \ref{Fig3a} and \ref{Fig3a1}.
\begin{figure}
\begin{centering}
\includegraphics[width=1\columnwidth]{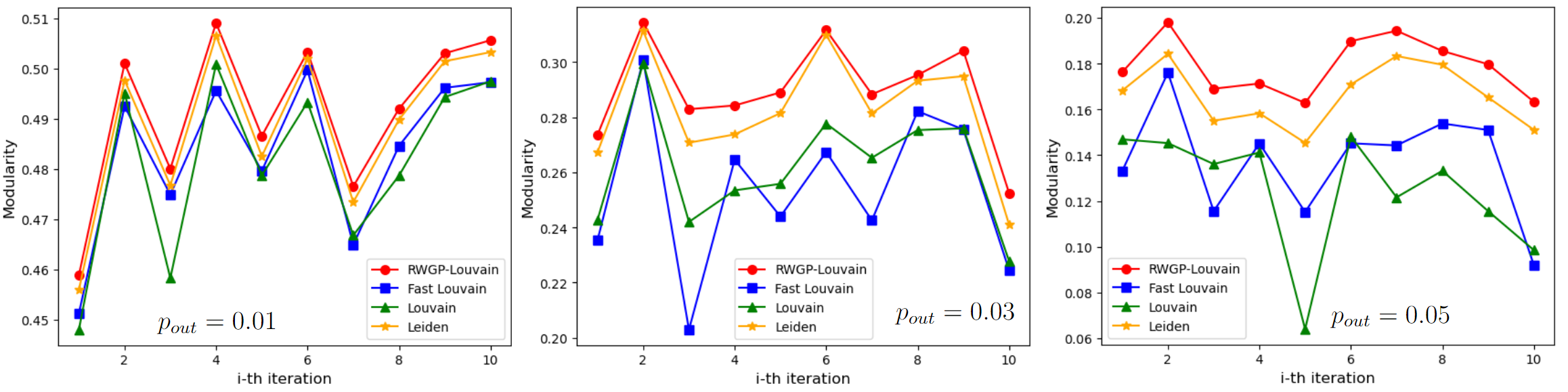}
\par\end{centering}
\caption{Modularity was observed in \textbf{Experiment 3} using the Gaussian random generator model with $N \in [4000;8000]$ and $m \in [50;70]$.}\label{Fig3a}
\end{figure}

\begin{figure}
\begin{centering}
\includegraphics[width=1\columnwidth]{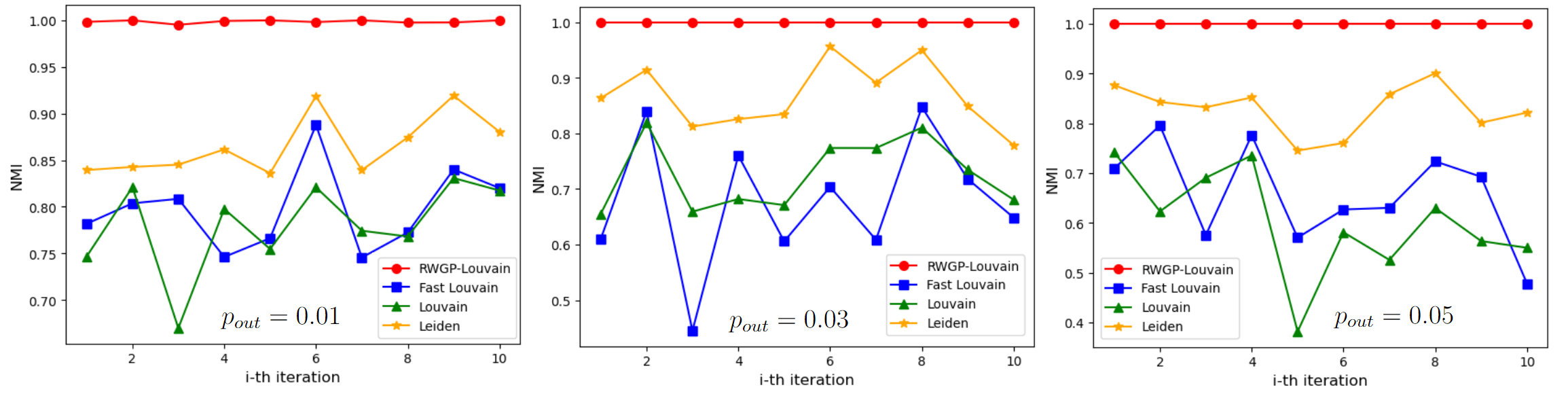}
\par\end{centering}
\caption{NMI was observed in \textbf{Experiment 3} using the Gaussian random generator model with $N \in [4000;8000]$ and $m \in [100;120]$.}\label{Fig3a1}
\end{figure}

\subsubsection{Experiments for RWGP-Louvain Algorithm on the random graph generated by the Planted-l partition
model}

With the Planted-l partition model, we also fix $p_{in}=0.7$ and explore $p_{out}$ values of $0.01$, $0.03$, and $0.05$, corresponding to graphs with clear community structure to graphs with unclear community structure.

\subsubsection*{Experiment 1: Random Graphs from the Planted-l Partition Model}
Using the Planted-l Partition Model, we set the number of communities ($l$) and the size of each community ($g$) are selected with a uniform distribution within the intervals $g \in [30;50]$ and $l \in [30;50]$. For RWGP-Louvain Algorithm, \(t\) is set to \(25\). The results are depicted in Figures \ref{Fig1b} and \ref{Fig1b1}.
\begin{figure}
\begin{centering}
\includegraphics[width=1\columnwidth]{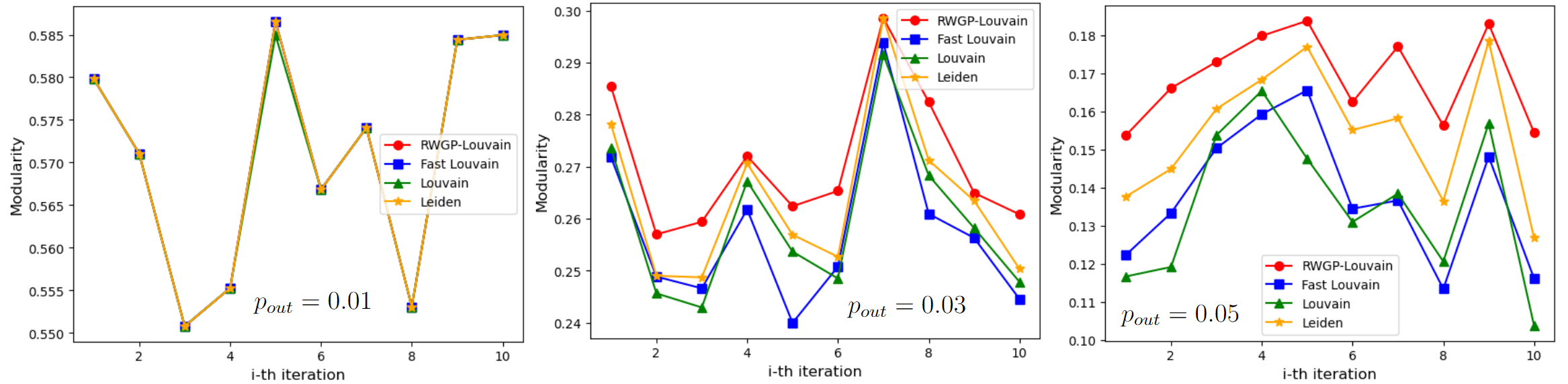}
\par\end{centering}
\caption{
Modularity was observed in \textbf{Experiment 1} using the Planted-l partition model with $g \in [30;50]$, $l \in [30;50]$.}\label{Fig1b}
\end{figure}

\begin{figure}
\begin{centering}
\includegraphics[width=1\columnwidth]{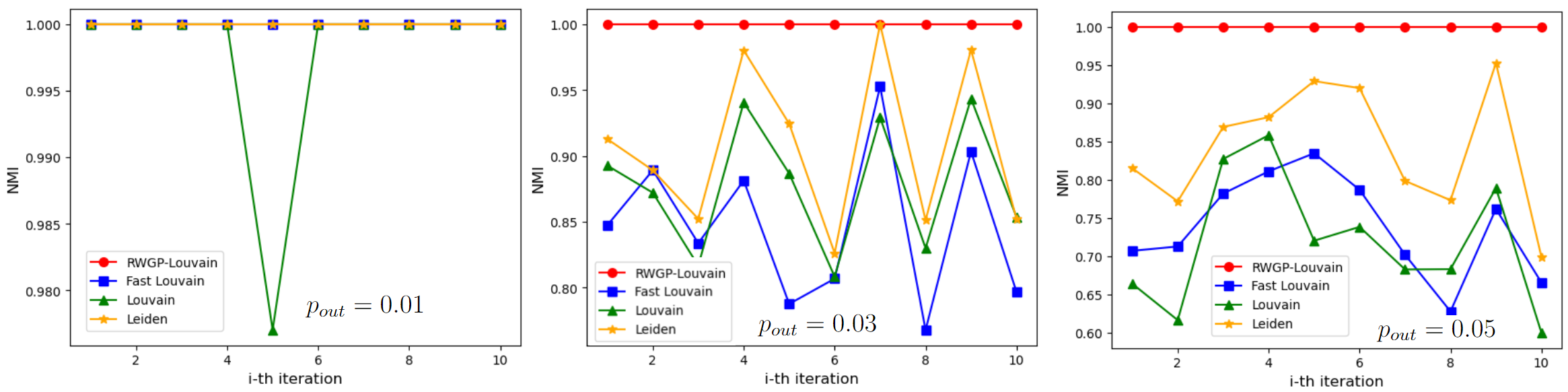}
\par\end{centering}
\caption{NMI was observed in \textbf{Experiment 1} using the Planted-l partition model with $g \in [30;50]$, $l \in [30;50]$.}\label{Fig1b1}
\end{figure}

\subsubsection*{Experiment 2:  Random Graphs from the Planted-l Partition Model}
Using the Planted-l Partition Model, we set the number of communities ($l$) and the size of each community ($g$) are selected with a uniform distribution within the intervals $g \in [50;70]$ and $l \in [50;70]$. For RWGP-Louvain Algorithm, \(t\) is set to \(25\). The results are illustrated in Figures  \ref{Fig2b} and \ref{Fig2b1}.
\begin{figure}
\begin{centering}
\includegraphics[width=1\columnwidth]{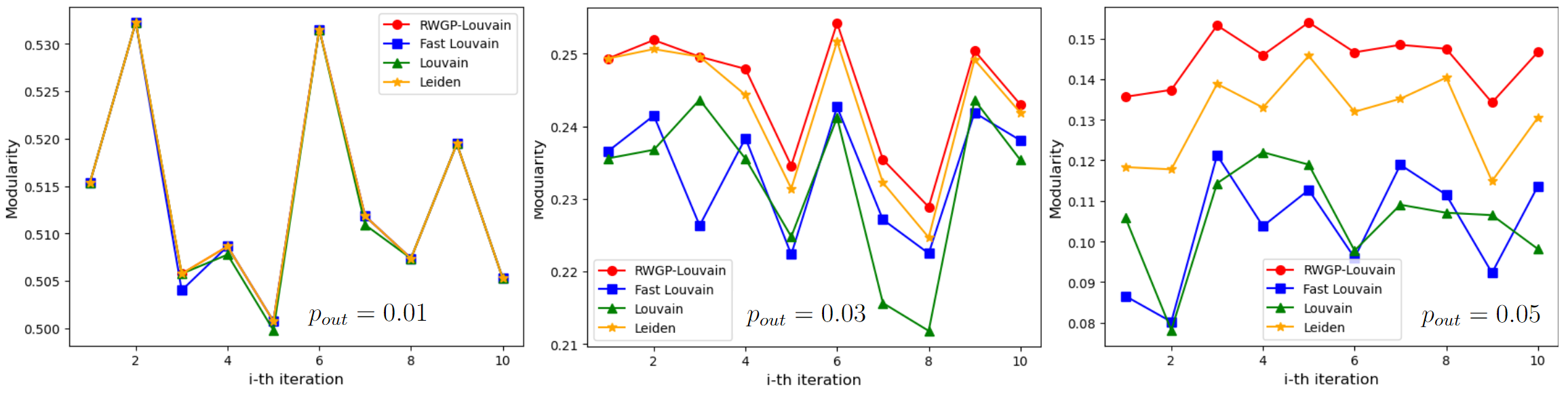}
\par\end{centering}
\caption{Modularity was observed in \textbf{Experiment 2} using the Planted-l partition model with $g \in [50;70]$, $l \in [50;70]$.}\label{Fig2b}
\end{figure}

\begin{figure}
\begin{centering}
\includegraphics[width=1\columnwidth]{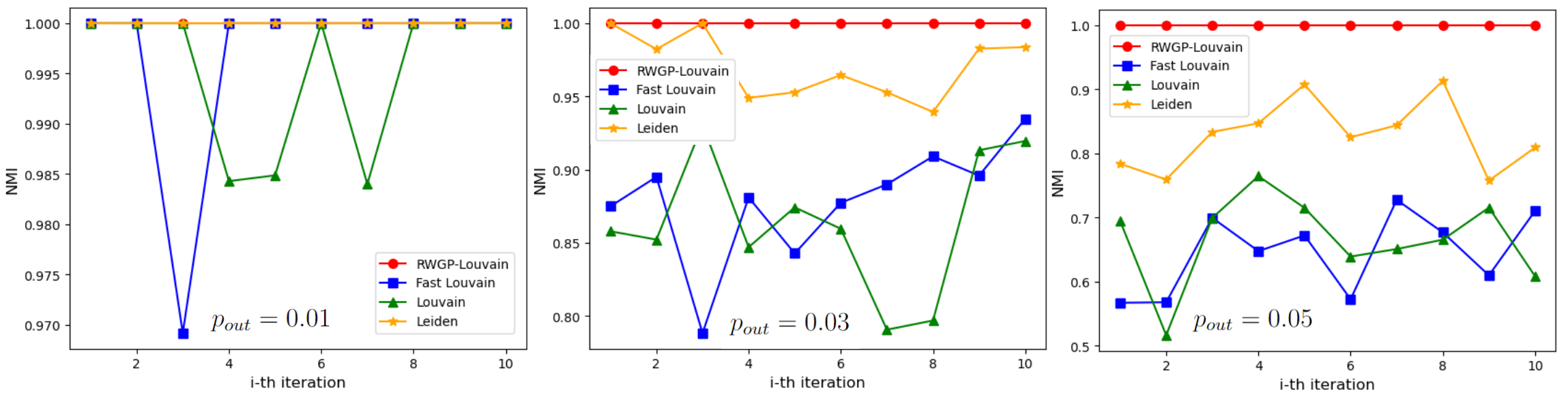}
\par\end{centering}
\caption{NMI was observed in \textbf{Experiment 2} using the Planted-l partition model with $g \in [50;70]$, $l \in [50;70]$.}\label{Fig2b1}
\end{figure}

\subsection{Experiments on real data}
Before performing the experiments, we will introduce some famous real data used in this section. \\
\textbf{Zachary's karate club: } Wayne W. Zachary examined a karate club's social network from 1970 to 1972, detailed in \cite{kara}. The network, featuring 34 members and capturing interactions beyond the club, gained prominence as a community structure example in networks following its analysis by Michelle Girvan and Mark Newman in 2002 \cite{Foot}.\\
 \textbf{College football: }
The college football network, examined in \cite{Foot}, serves as a benchmark for community detection. It illustrates the games played by Division I colleges in the autumn of 2000, with each node representing a football team and each edge representing a regular season game. With 115 nodes and 616 edges, the network can be divided into 12 communities based on athletic conferences, each comprising 8 to 12 teams.\\
\textbf{Jazz network:}
Data was sourced from The Red Hot Jazz Archive digital database \cite{Ncite7}, comprising 198 bands active between 1912 and 1940, predominantly in the 1920s. The database identifies musicians in each band, but distinguishing their temporal involvement is challenging, hindering the study of the collaboration network's temporal evolution. The remaining 1275 musicians' names are dispersed across the bands.\\
\textbf{Metabolic network:}
As in  \cite{meta}, a metabolic network encompasses the entire metabolic and physical processes governing a cell's physiological and biochemical characteristics. It includes metabolic reactions, pathways, and the regulatory interactions orchestrating these reactions.

In addition, we will also perform experiments on the following real data:  Hamster households, hamster friendships, and Asoiaf. These data we can see in \cite{hams}.

Now, we will conduct experiments to compare the effectiveness of our proposed Random Walk Graphs Partition Louvain Algorithm (RWGP-Louvain) with existing algorithms, namely the original Louvain algorithm \cite{Louvain}, the Fast Louvain algorithm \cite{FastLouvain}, and the Leiden algorithm \cite{leiden} on real data.  After obtaining the clustering results, we will calculate modularity (using the formula \ref{QQ}) and record the modularity results in the following table \ref{tab1}.

\begin{table} 
\begin{centering}
\begin{tabular}{|p{4cm}|p{2cm}|p{3cm}|p{2cm}|p{2.5cm}|}
\hline Graph $G=(|V|,|E|)$ & Louvain & RWGP Louvain & Leiden & Fast Louvain \tabularnewline
\hline 
 Karate \cite{meta}, $G = ( ,)$ & 0.2394806   & 0.2450690  &  0.2404668&  0.2404668    \tabularnewline
 \hline 
 Metabolic network \cite{meta} \\ $G = ( 453 , 2025 )$ &  0.2960217   & 0.3054383  & 0.2965548&  0.29486578    \tabularnewline
\hline 
  College football \cite{Foot}\\ $G = ( 115, 613 )$& 0.5357494 & 0.5400073  & 0.5357494 & 0.5316192    
  \tabularnewline
\hline 
  Jazz network \cite{Ncite7}\\ $G = ( 198, 2742 )$ &  0.2833862 & 0.2850413 &  0.2828381 & 0.2821492  \tabularnewline
\hline 
 Hamster households \cite{hams}\\ $ G = ( 921, 4032 )$ &  0.1678505 & 0.2058422  &  0.1983055 &  0.1940693 
 \tabularnewline
\hline 
Hamsters friendships \cite{hams}\\ $G = ( 1858, 12534 )$  & 0.3096735 & 0.3308066   & 0.3216459 &  0.3095751
\tabularnewline
\hline 
Asoiaf \cite{hams}\\ $G = ( 796 , 2823 )$ &  0.5182439  & 0.5370974 &0.5404249 & 0.5185884  \tabularnewline
\hline 
 
\end{tabular}
\par\end{centering}
\caption{In this table, we present the values of Modularity (using the formula \ref{QQ}) corresponding to the clustering results of the algorithms.}\label{tab1}

\end{table}
\subsection{Conclusion of the experiments}

The above results show that our algorithms are efficient in almost all experiments.
\begin{itemize}
\item In the Subsection \ref{RWGP}, we perform some simple experiments
using two models Gaussian random generator model and Planted-l
partition model to compare our Random Walk Graphs Partition Algorithm 1 and Random Walk Graphs Partition Algorithm 2 with Louvain algorithm \cite{Louvain}, Newman's Spectral Method \cite{main}. We see that our algorithm works well on graphs with clear community structures. On graphs with unclear structures, our Random Walk Graphs Partition Algorithm 2 algorithm still works more efficiently and has less computational complexity than Newman's Spectral Method.

\item In Subsection \ref{RWGPL}, we performed experiments comparing the effectiveness of our proposed Random Walk Graphs Partition Louvain Algorithm and the Louvain ucite{Louvain}, Fast Louvain Algorithm \cite{ FastLouvain} and Leiden algorithm \cite{leiden} through graphs on randomly generated graphs using Gaussian random generator model and Planted-l partition model. The algorithms we investigated on graphs with apparent community structures produce good results. Our Random Walk Graphs Partition Louvain Algorithm performs much better than the Louvain, Fast Louvain, and Leiden algorithms for graphs with unclear community structures.

\item Furthermore, when performing experiments using real data, our Random Walk Graphs Partition Louvain Algorithm is also more effective than the remaining algorithms on some real data.
\end{itemize}

\section{Conclusion and further work}

This paper introduced a novel approach to community detection through graph partitioning by leveraging a random walk strategy akin to the spectral method presented in \cite{main}, referred to as the Random Walk Graph Partition Algorithm. Our proposed algorithm demonstrates superior computational efficiency compared to the spectral algorithm discussed in \cite{main}.

Moreover, recognizing the Random Walk Graph Partition Algorithm as a phase within the Louvain algorithm, we introduce a new algorithm called the Random Walk Graph Partition Louvain Algorithm. This algorithm retains computational complexity equivalent to the Louvain algorithm while achieving enhanced efficiency, particularly in scenarios involving graphs with ambiguous community structures.

To validate the rationale and efficacy of our proposed algorithms, we conduct experiments on randomly generated and real-world datasets. The results underscore the effectiveness of our approaches, reinforcing their applicability in practical community detection scenarios.

\section*{Acknowledgments}
This research was supported 
by the International Centre of Research and Postgraduate Training in Mathematics, Institute of Mathematics, project code: ICRTM04-2021.01, and 
Vingroup Innovation Foundation, project code: VINIF.2023.STS.53.

\end{document}